\documentclass[10pt]{article}
\usepackage{enumitem}
 \usepackage{amsfonts}
\usepackage{amsmath}
\usepackage{graphicx}
\usepackage{fancyhdr}
\usepackage{amsthm}
\usepackage{amssymb}
\usepackage{amscd}
\usepackage{float}
\usepackage{xcolor}
\newtheorem{definition}{Definition}[section]

\newtheorem{lemma}{Lemma}[section]
\newtheorem{proposition}{Proposition}[section]
\newtheorem{remark}{Remark}[section]

\begin{document}
\title{{Equilibrium price and optimal insider trading strategy under stochastic liquidity with long memory }\thanks{This work was supported by  the National Natural Science Foundation of China (11471230, 11671282).}}
%\title{{Analytical form pricing formulas for discretely sampled volatility swaps under L\'{e}vy process with stochastic volatility}\thanks{}}
\author{{ Ben-zhang Yang$^a$, Xinjiang He$^b$, Nan-jing Huang$^{a,}$\thanks{Corresponding author. E-mail address: nanjinghuang@hotmail.com}}\\
{\small\it a. Department of Mathematics, Sichuan University, Chengdu, Sichuan 610064, P.R. China}\\
{\small\it b. School of Mathematics and Applied Statistics, University of Wollongong NSW 2522, Australia}
}
\date{}
\maketitle
\vspace*{-9mm}
\begin{center}
\begin{minipage}{5.8in}
{\bf Abstract:}
In this paper, the Kyle model of insider trading is extended by characterizing the trading volume with long memory and allowing the noise trading volatility to follow a general stochastic process. Under this newly revised model, the equilibrium conditions are determined, with which the optimal insider trading strategy, price impact and price volatility are obtained explicitly. The volatility of the price volatility appears excessive, which is a result of the fact that a more aggressive trading strategy is chosen by the insider when uninformed volume is higher. The optimal trading strategy turns out to possess the property of long memory, and the price impact is also affected by the fractional noise.
\\ \ \\
{\bf Key Words:}
Asset pricing; equilibrium; insider trading; optimal investment; liquidity; fractional Brownian motion.
\\ \ \\
%\textbf{Classification:}

\end{minipage}
\end{center}
\section{Introduction}

Equilibrium asset pricing is always an important and ongoing topic for any financial market, and it has been extensively studied by a number of authors \cite{Chang05,Collin16,Guasoni,Liu15,Yang17}. It also can be used to deal with insider trading problem and the history can be dated back to 1985, when Kyle \cite{Kyle85} developed a model, in which a large trader is assumed to possess long-lived private information about the value of a stock that will be revealed at some known date and optimally trades continuously into the stock to maximize his/her expected profits, while risk-neutral market makers try to infer the information possessed by the insider from the aggregate order flow. The resulted equilibrium price dynamic actually responds linearly to the order flow instead of being fully revealed, since the order flow is also driven by uninformed noise traders aiming solely at liquidity purposes. Albeit appealing, the Kyle model also suffers from two major drawbacks; Kyle's lambda, which is a measurement of the equilibrium price impact of the order flow is constant, and the price volatility turns out to be constant and independent of the noise trading volatility.

To further investigate the impact of asymmetric information on asset prices, volatility, volume, and market liquidity, Admati and Pfleiderer \cite{Admati88}, Foster and Viswanathan \cite{Foster90,Foster93} were the first to modify the Kyle model, extending it to dynamic economies with myopic agents so that the price volatility is no longer constant since new private information is introduced every period caused by the deterministic changes in the noise trading volatility. Stochastic variation was also introduced in the noise trading volume by Foster and Viswanathan \cite{Foster95}, and this model was then empirically tested to show the joint behavior of the price volatility, volume, and price impact. However, a noticeable shortcoming of these extensions is the assumption that the information is short-lived, which is not realistic. Later on, deterministic noise trading volatility was directly incorporated into the Kyle model so that intra-day patterns of liquidity trading can be successfully captured \cite{Back98}. Such an assumption on the noise trading volatility is still not appropriate since it gives rise to the deterministic price impact and further leads to the situation where expected execution costs of liquidity traders do not depend on the timing of their trading. A more recent extension to the Kyle model was proposed by Collin-Dufresne and Fos \cite{Collin16}, who considered stochastic noise trading volatility, and showed that this is the only guarantee that insider's liquidity timing option can generate a stochastic price volatility as well as a nonzero correlation among the price volatility, market depth, and uninformed trading volume.

Despite all the advantages of the model proposed in \cite{Collin16}, it has not taken into consideration that the volatility is characterized by long memory, which is a consensus among financial econometricians. In particular, Bollerslev and Jubinski \cite{Bollerslev99} presented the evidence of long memory in the volatility and investigated the extent to which the volume and volatility share common long-run dependencies. In line with this, they also proposed that the mixture-of-distributions hypotheses (MDH) might be better viewed as long-run proposition, after observing that occasional structural breaks can give rise to long-range dependence \cite{Granger04}. This implies long-range dependence could be induced under the MDH by occasional changes in the mean and/or volatility of the intra-day volumes and returns, which are generated due to the reaction traders make to information events. Furthermore, fractionally-integrated models were shown to provide a useful description of volatility dynamics in the presence of structural breaks because they effectively allow the unconditional variance to change slowly over time \cite{Diebold01}, and Hyung et al. \cite{Hyung06} further demonstrated that fractionally-integrated models provide the best volatility forecasts, when it is impossible to identify the volatility breaks before they occur. Very recently, a new fractional stochastic process for the volatility was proposed in \cite{Yue17,Yue18}, in which the general fractional stochastic models are shown to better capture dynamic volatilities of the assets and changing correlations between the returns and volatilities of the underlying assets.

Considering the amount of resources spent by market participants to separate the component of the informed order flow from the uninformed one as well as the long memory property exhibited by the volatility, it is very demanding to understand how these affect the equilibrium price and liquidity, since similar to most of quantitative finance/economics area \cite{He16, Zhu18}, finding an appropriate model that reflects the most characteristics of the real market is vital to perform correct analysis. To incorporate these, in this paper, the Kyle model is generalized so that the noise trading volatility is allowed to evolve under a stochastic setting, and the trading volume is governed by a fractional stochastic dynamic system perturbed by a memory noise. The main contributions of this paper can be summarized into four aspects. Firstly, a new general framework involving the fractional stochastic environment is established for the insider trading problem, and the corresponding equilibrium price is shown to follow a new class of bridge processes that converge almost surely at maturity to the ex ante value, of which only the insider has the knowledge. Secondly, the derived optimal trading strategy displays long memory, being proportional to the undervaluation of the asset with the corresponding rate being an increasing and decreasing function of the current state of liquidity and price impact, respectively. Thirdly, the price impact proves to be a submartingale, being negatively correlated with the noise trading volatility, and it is also affected by the fractional noise. Lastly, the rate of price discovery in our new model increases with the noise trading volatility, while an apposite trend can be observed when the price impact is taken into consideration.

The rest of this paper is organized into five sections. Section 2 presents some necessary preliminaries. In Section 3, we introduce the general model and the details on achieving an equilibrium are illustrated. In Section 4, two specific cases are considered to emphasize key features of the equilibrium and possible extensions are discussed followed by some concluding remarks are given in the last section.

\section{Preliminaries}
In this section, some useful concepts and results are presented. We assume that $(\Omega,\mathcal{F}_T,P)$ is a probability space, with $\mathcal{F}_t\triangleq \sigma(X_s;0\leq s \leq t)$ being the smallest $\sigma$-filed with respect to which the random variable $X_s$ is measurable for all $s \in [0,t]$. Let $L^{\infty}(\Omega,\mathcal{F}_T,P)$ represent the space of all bounded and $\mathcal{F}$-measurable random variables and $\mathcal{M}^{c,loc}$ denote the space of all continuous local martingales. Assume that $\langle X\rangle_t$ is defined as the quadratic variation process of $X_t$ and $C_1$ stands for the space of the continuous functions $x=(x_t;0\leq t \leq 1)$ on $[0,1]$. We also let $\mathcal{P}$ and $\mathcal{P}\bigotimes \mathcal{B}(R^{d+1})$ be the predictable $\sigma$-field and the product $\sigma$-filed formed from the $\sigma$-fields $\mathcal{P}$ and $\mathcal{B}(R^{d+1})$, respectively.

\begin{lemma}\label{strong0}(Theorem 4.6 in Chapter 3 of \cite{Karatzas})
Let $M=\{M_t,\mathcal{F}_t;0\leq t \leq \infty\}\in \mathcal{M}^{c,loc}$ satisfy
$\lim_{t \rightarrow \infty}\langle M\rangle_t=\infty$, $a.s.$ Define, for each $0\leq s<\infty$, the stopping time as
$$T(s)=\inf\{t\geq0;\langle M\rangle_t>s \}.$$
Then the following time-changed process
$$B_s\triangleq M_{T(s)}, \quad \mathcal{G}_s\triangleq \mathcal{F}_{T(s)}; \quad 0\leq s <\infty$$
is a standard Brownian motion. In particular, the filtration $\{\mathcal{G}_s\}$ satisifies the usual conditions and
$$M_t=B_{\langle M\rangle_t};\quad 0\leq t<\infty, \quad a.s.$$
\end{lemma}

\begin{lemma}\label{strong}(Problem 9.3 in Chapter 2 of \cite{Karatzas})
Let $W=\{W_t,\mathcal{F}_t; 0<t<\infty\}$ be a standard Brownian motion. Then
$$\lim_{t\rightarrow \infty}\frac{W_t}{t}=0, \quad a.s.$$
\end{lemma}

\begin{lemma}\label{sde1}(Theorem 4.6 in \cite{Liptser01})
Let the non-anticipative functionals $a(t,x)$, $b(t,x)$, $t\in[0,1]$, $x\in C_1$, satisfy the Lipschitz condition
\begin{equation*}
\begin{aligned}
|a(t,x)-a(t,y)|^2+|b(t,x)-b(t,y)|^2\leq L_1 \int_{0}^{t}|x_s-y_s|^2dK(s)+L_2|x_t-y_t|^2
\end{aligned}
\end{equation*}
and
\begin{equation*}
\begin{aligned}
a^2(t,x)+b^2(t,x)\leq L_1 \int_{0}^{t}(1+x_s^2)dK(s)+L_2(1+x_t^2),
\end{aligned}
\end{equation*}
where $L_1$ and $L_2$ are constants, $K(s)$ is a non-decreasing right continuous function with $0<K(s)<1$, and $x,y\in C_1$. Let $\eta=\eta(\omega)$ be an $\mathcal{F}_0$-measurable random variable with $P(|\eta(\omega)<\infty|)=1$. Then the equation
$$dx_t=a(t,x)dt+b(t,x)dW_t, \quad x_0=\eta,$$
has a unique strong solution $\xi=(\xi_t,\mathcal{F}_t)$.
\end{lemma}

Let $\ell: R \rightarrow R_+$ be a strictly positive and continuous function. We say
that $\ell$ belongs to the class $L$ if it satisfies: for all $-\infty<a\leq0\leq b<+\infty$, the following ODEs
\begin{equation}\label{ODE1}
L_t=a-\int_{t}^{T}\ell(L_s)ds, \quad U_t=b+\int_{t}^{T}\ell(U_s)ds
\end{equation}
have global bounded solutions on $[0,T]$.

\begin{lemma}\label{Lea}(\cite{Lepeltier97})
 $\ell\in L$ if and only if
$$\int_{-\infty}^{0}\frac{dx}{\ell(x)}=\int_{0}^{\infty}\frac{dx}{\ell(x)}=\infty.$$
Moreover, when $\ell\in L$, the equations presented in \eqref{ODE1} have unique solutions.
\end{lemma}

\begin{lemma}\label{Leb}(\cite{Lepeltier97})
Assume that $\xi \in L^{\infty}(\Omega,\mathcal{F}_T,P)$ and $f$ is a $\mathcal{P}\bigotimes \mathcal{B}(R^{d+1})$ measurable function such that  $f(t,\omega,\cdot,\cdot)$ is continuous for all $t$ and $\omega$, and there exists some finite constant $C$ such that
$$|f(t,\omega,y,z)|\leq \ell(y)+C|z|^2$$
for all $t$, $\omega$, $y$ and $z$. If $\ell \in L$, then the BSDE (backward stochastic differential equation)
\begin{equation}\label{BSDE}
Y_t=\xi+\int_{t}^{T}f(s,\omega,Y_s,\Lambda_s)ds-\int_{t}^{T}\Lambda_sdW_s
\end{equation}
has a maximal bounded solution $(Y,\Lambda)$. Moreover, $Y$ is a continuous process and $L_0\leq L_t\leq Y_t(\omega) \leq U_t \leq U_0$ holds for all $t$ and $\omega$, where $(L,U)$ are the unique solutions of \eqref{ODE1} with $b=\|\xi\|_\infty$ and $a=-\|\xi\|_\infty$.

\end{lemma}

\begin{lemma}\label{Lec}(\cite{Lepeltier97})
Assume that $f$ and $h$ are two $\mathcal{P}\bigotimes \mathcal{B}(R^{d+1})$ measurable functions such that
\begin{itemize}
\item[(i)] for all $t$, $\omega$; $f(t,\omega,\cdot,\cdot)$ is continuous;
\item[(ii)] for all $t$, $\omega$, $y$, $z$; $f(t,\omega,y,z)\geq h(t,\omega,y,z)$;
\item[(iii)] there exists a constant $C>0$ such that, for all $t$, $\omega$, $y$, $z$; $|f(t,\omega,y,z)|<\ell(y)+C|z|^2$ where $\ell \in L$;
\item[(iv)] $\xi$, $\eta$ $\in$ $L^{\infty}(\Omega,\mathcal{F}_T,P)$ and $\xi \leq \eta$ $P-a.s.$
\end{itemize}
If $(X,Z)$ is a maximal bounded solution of \eqref{BSDE} with terminal value $\xi$ and coefficient $f$,  and $(Y,\Gamma)$ is a bounded solution of \eqref{BSDE} with terminal value $\eta$ and coefficient $h$, then $Y\leq X$.
\end{lemma}

\begin{definition}(\cite{Mandelbrot})\label{fracB}
  A fractional Brownian motion $B_t^H$ with Hurst index $H$ is a centered Gaussian process such that its covariance function $R(t,s)=E[B_t^HB_s^H]$ is given by
  $$R(t,s)=\frac{1}{2}(|t|^{2H}+|s|^{2H}-|t-s|^{2H}),$$
  where $0<H<1$.
\end{definition}

\begin{remark}
If $H=\frac{1}{2}$, $R(t,s)=\min(t,s)$ and $B_t^H$ is the usual standard Brownian motion.
If $H\neq\frac{1}{2}$, $B_t^H$ is neither a semimartingale nor a Markov process. It is a process of long memory in the following sense \cite{Shiryaev99}: If $\rho_n=E[B_1^H(B_{n+1}^H-B_n^H)]$, then the series $\sum_{n=0}^{\infty}\rho_n$ is either divergent or convergent with very late rate.
\end{remark}

\section{Equilibrium price and optimal strategy}

We start by presenting several fundamental assumptions made in our model. For a certain stock price process $P_t$, there is an insider who is risk-neutral and knows its terminal stock value $v$. The insider tries to maximize the expectation of his/her terminal profit as follows:
\begin{equation}\label{max}
J_t=\max_{\{\theta_s\}_{s\geq t}\in \mathcal{A}} \mathrm{E}\left[\left.\int_{t}^{T}(v-P_s)\theta_s ds \right| \mathcal{F}^Y_t,v\right].
\end{equation}
Here, following \cite{Back98}, we assume that the insider chooses an absolutely continuous trading rule $\theta$ which belongs to an admissible set
\begin{equation}\label{e2}
\mathcal{A}=\left\{\theta\left|\mathrm{E}\left[\int_{0}^{T}\theta_s^2 ds\right]< \infty\right.\right\}.
\end{equation}
Moreover, the aggregate order flow arrival $Y_t$ (caused by the insider's demand and the noise-trader's demand) is assumed to follow
\begin{equation}\label{YYY}
dY_t=\theta_tdt+\sigma_t dB^{H}_t,
\end{equation}
where $B^{H}_t$ is a fractional Brownian motion independent of $v$, which implies that the volume has remarkably memory characteristics, $H$ is restricted within $(\frac{1}{2}, 1]$ as it is of interest in finance to investigate the effect of the long-range dependence exhibited by the volatility, $\mathcal{F}^Y_t$ is the filtration generated by the historic information of the aggregate order flow $Y^t=\{Y_u\}_{u\leq t}$. Another assumption we made here is that that $\sigma_t$ follows
\begin{equation}\label{MMT}
d\sigma_t=m(t,\sigma^t)\sigma_tdt+\nu(t,\sigma^t)\sigma_tdM_t,
\end{equation}
where $M_t$ is a martingale, the long term rate $m(t,\sigma^t)$ and volatility $\nu(t,\sigma^t)$ are dependent on the past history of the volatility $\sigma^t$, instead of depending on the history of $Y$.

On the other hand, the market maker is also risk-neutral, and has a prior information that $v$ follows the normal distribution $N(P_0,\Sigma_0)$, instead of having the knowledge of the terminal value $v$. The market maker absorbs the total order flow by trading against it at a price which is determined to achieve break even on average. As the market maker is risk-neural, equilibrium break-even requires that the market clearing price is
\begin{equation}\label{Pt}
P_t=\mathrm{E}\left[v|\mathcal{F}^Y_t\right].
\end{equation}

Finally, both the market maker and the insider are assumed to perfectly observe the history of $\sigma$, which indicates that the filtration $\mathcal{F}^Y_t$ here contains both history of the order flow $Y^t$ and volatility $\sigma^t$. This is different from what is assumed in the Kyle model, where only equilibrium prices are observable, since the insider there is able to recover the total order flow, while only observing prices are not sufficient for the recovery of the noise trading volatility, as the volatility of the uninformed order flow in our model is no longer constant, but a stochastic variable.

It should be remarked here that if the uncertain term were general Brownian motion without memories, then the model would degenerate to the one studied by Collin-Duersne and Fos \cite{Collin16}. Furthermore, if the noise trading $\sigma_t$ were constant, then the model would become exactly the same as the classical Kyle-Back model \cite{Kyle85}.

\subsection{Equilibrium price process}
\begin{definition}\label{def1}
An equilibrium occurs when $(P_T,\theta_t)$ satisfy the condition \eqref{Pt} while solving the insider's optimal problem \eqref{max}.
\end{definition}

To solve this equilibrium, we proceed in a three-step process:  a) the stock price dynamic is presented, being consistent with the market maker's risk-neutral filtration, which is conditional on a conjectured strategy rule followed by the insider; b) the insider's optimal problem \eqref{max} is solved under the assumed dynamic \eqref{Pt}; c) solving the conjecture rule \eqref{Pt} is shown to be consistent with the optimal solution of \eqref{max}, which shows that we have reached the equilibrium defined in Definition \ref{def1}.

We start by presenting a few lemmas, which will be used when deriving the optimal trading strategy of the insider. Lemma \ref{Ltheta} below shows that the change of the stock price process $P_t$ is linear in the order flow, if the insider takes a trading strategy that is linear in his/her profit with memory.

\begin{lemma}\label{Ltheta}
If the insider adopts the following memorable trading strategy
\begin{equation}\label{stra}
\theta_t=\beta_t(v-P_t)-(H-\frac{1}{2})\psi_t
\end{equation}
for some $\mathcal{F}_t^Y$ adapted process $\beta_t$ with
$$\psi(t)=\int_{0}^{t}(t-u+\varepsilon)^{H-\frac{3}{2}}dW_u$$
for any $\varepsilon>0$, then the stock price defined in \eqref{Pt} follows
\begin{equation}\label{Pt2}
dP_t=\lambda_tdY_t,
\end{equation}
where $\lambda_t$ is the price impact defined by
\begin{equation}\label{lambdat}
\lambda_t=\frac{\beta_t\Sigma_t}{\varepsilon^{2H-1} \sigma^2(t)},
\end{equation}
and $\Sigma_t$ is the following conditional variance
\begin{equation}\label{Sigma}
\Sigma_t=\mathrm{E}\left[\left.(v-P_t)^2\right|\mathcal{F}_t^Y\right].
\end{equation}
Moreover, the dynamic of $\Sigma_t$ can be formulated as follows:
\begin{equation}\label{Sigma9}
d\Sigma_t=-\lambda_t^2\varepsilon^{2H-1}\sigma_t^2dt.
\end{equation}
\end{lemma}
\begin{proof}
If the stochastic process $\psi_t$ is defined by
$$\psi_t=\int_{0}^{t}(t-u+\varepsilon)^{H-\frac{3}{2}}dW_u,$$
then the stochastic theorem of Fubini yields
\begin{eqnarray}
\int_0^t \psi_s ds&=&\int_{0}^{t}\int_{0}^{s}(s-u+\varepsilon)^{H-\frac{3}{2}}dW_uds\nonumber\\
&=&\int_{0}^{t}\left[\int_{s}^{u}(s-u+\varepsilon)^{H-\frac{3}{2}}ds \right]dW_u\nonumber\\
&=&\frac{1}{H-\frac{1}{2}}\left[\int_{0}^{t}(t-u+\varepsilon)^{H-\frac{1}{2}}dW_u- \varepsilon^{H-\frac{1}{2}} W_t \right].\nonumber
\end{eqnarray}
Define
$$B_t^{\varepsilon,H}=\int_{0}^{t}(t-s+\varepsilon)^{H-\frac{1}{2}}dW_s.$$
Then, based on the expression of the fractional Brownian motion
$$B_t^H=\int_{0}^{t}(t-s)^{H-\frac{1}{2}}dW_s,$$
we can obtain
\begin{equation}\label{e12}
\int_0^t \psi_s ds=\frac{1}{H-\frac{1}{2}}\left[B_t^{\varepsilon,H}-\varepsilon^{H-\frac{1}{2}} W_t\right].
\end{equation}
This leads to
$$B_t^{\varepsilon,H}=(H-\frac{1}{2})\int_{0}^{t}\psi_sds+\varepsilon^{H-\frac{1}{2}} W_t,$$
and thus $B_t^{\varepsilon,H}$ is a semimartingale.
Moreover, $B_t^{\varepsilon,H}$ uniformly converges to $B_t^H$ for $t\in [0,T]$ in $L^2(\Omega)$ when $\varepsilon$ tends to $0$. Therefore, we hereafter use $B_t^{\varepsilon,H}$ to approximate the original fractional Brownian motion $B_t^H$, which is a common practice in mathematical finance now \cite{Mrazek, Thao}.

Now, with the standard Gaussian projection theorem (Theorems 12.6 and 12.7 in \cite{Liptser01}), we can derive
\begin{eqnarray}
P_{t+dt}&=&\mathrm{E}[v|Y^t,Y_{t+dt},\sigma^t,\sigma_{t+da}]\nonumber\\
&=&\mathrm{E}[v|Y^t,\sigma^t]+\frac{\mathrm{Cov}(v,Y_{t+dt}-Y_t|Y^t,\sigma^t)}{\mathrm{V}(Y_{t+dt}-Y_t|Y^t,\sigma^t)}\nonumber\\
&& \mbox{}\times \left(Y_{t+dt}-Y_t-\mathrm{E}[Y_{t+dt}-Y_t|Y^t,\sigma^t]\right)\nonumber\\
&=&P_t+\frac{\mathrm{E}[(v-P_t)(\theta_t+(H-\frac{1}{2})\psi_t)|Y^t,\sigma^t]dt}{\varepsilon^{2H-1}\sigma^2_tdt}(Y_{t+dt}-Y_t)\nonumber\\
&=& P_t+\frac{\mathrm{E}[\beta_t(v-P_t)^2|Y^t,\sigma^t]dt}{\varepsilon^{2H-1}\sigma^2_tdt}(Y_{t+dt}-Y_t)\nonumber\\
&=&P_t+\frac{\beta_t\Sigma_t}{\varepsilon^{2H-1}\sigma^2_t}dY_t.
\end{eqnarray}
Here, the second equality uses the fact that $\sigma^t$ is independent of the asset value in the order flow, the third equality is obtained from the fact that the expected change in the order flow is zero for the conjectured policy $\theta_t$, and the last equality follows from \eqref{Sigma}.

Finally, applying the projection theorem yields
\begin{equation}
\mathrm{Var}[v|Y^t,Y_{t+dt},\sigma^t,\sigma_{t+da}]=\mathrm{Var}[v|Y^t,\sigma^t]-\left(\frac{\beta_t\Sigma_t}{\varepsilon^{H-\frac{1}{2}}\sigma_t}\right)^2\mathrm{Var}[Y_{t+dt}-Y_t|Y^t,\sigma^t],
\end{equation}
which leads to
\begin{equation}
\Sigma_{t+dt}=\Sigma_t-\varepsilon^{2H-1}\lambda_t^2\sigma_t^2,
\end{equation}
where $\lambda_t$ is defined by \eqref{lambdat}.
This completes the proof.
\end{proof}

As our model now exhibits long memory, we also need to introduce the new market depth process and derive the corresponding equilibrium price process, before we are able to find the optimal solution to the problem \eqref{max}.

\begin{lemma}\label{G001}
Define $G_t$ by setting $\lambda_t=\sqrt{\frac{\Sigma_t}{G_t}}$ and assume that the market depth (i.e. Kyle's lambda) process $\frac{1}{\lambda_t}$ is martingale. If $\Sigma_T=0$,  $\sigma_t$ is uniformly bounded above by $\overline{\sigma}$ and below by $\underline{\sigma}>0$, then there exists a unique bounded solution $G_t$ satisfying
\begin{equation}\label{bound1}
\underline{\sigma}^2\varepsilon^{2H-1}(T-t)\leq G_t \leq \overline{\sigma}^2\varepsilon^{2H-1}(T-t).
\end{equation}
\end{lemma}
\begin{proof}
By setting $\mathrm{E}[d\frac{1}{\lambda_t}]=0$ and using \eqref{Sigma9},
one has
\begin{equation}
\mathrm{E}[dG_t]=-\frac{\varepsilon^{2H-1}\sigma_t^2}{2\sqrt{G_t}}dt
\end{equation}
with the terminal condition $G_T=0$, which implies that $G_t$ satisfies
\begin{equation}\label{GGT}
\sqrt{G_t}=\mathrm{E}\left[\left.\int_{t}^{T}\frac{\varepsilon^{2H-1}\sigma_s^2}{2\sqrt{G_s}}ds\right|\mathcal{F}_t^\sigma\right].
\end{equation}
Thus, $\sqrt{G_t}$ solves the BSDE
$$dy_t=-f(t,y_t)dt-\Lambda_tdM_t,\quad y_T=0,$$
where
$$
 f(t,y)=\frac{\varepsilon^{2H-1}\sigma_t^2}{2y}.
$$
If we define $\ell(y)=\frac{\varepsilon^{2H-1}\overline{\sigma}^2}{2|y|}$, then it is not difficult to find that $f(t,y_t)\leq \ell(y_t)$ for all $(t,\omega)$ and
$$
\int_{0}^{\infty}\frac{1}{\ell(x)}dx=\int_{-\infty}^{0}\frac{1}{\ell(x)}dx.
$$
As a result, according to Lemmas \ref{Lea} and \ref{Leb}, there exist
two solutions $L(t)$ and $U(t)$ solving $L_t=-\int_{t}^{T}\ell(L_s)ds$ and $U_t=-\int_{t}^{T}\ell(U_s)ds$, respectively,  such that $L_t \leq G_t \leq U_t$. Moreover, it is easy to derive that
$$
U_t^2=-L_t^2=\varepsilon^{2H-1}\overline{\sigma}^2(T-t),
$$
yielding the upper bound of $G_t$. On the other hand, if we consider the solution to the following backward equation
$$dx_t=-\frac{\varepsilon^{2H-1}\underline{\sigma}^2}{2x_t}dt-\widetilde{\Lambda}_tdM_t$$
with terminal condition $x_T=0$, then we can actually obtain
$$
x_t=\varepsilon^{H-\frac{1}{2}}\underline{\sigma}\sqrt{T-t}
$$
by setting $\widetilde{\Lambda}_t=0$. Considering
$$
f(t,y_t)\geq \frac{\varepsilon^{2H-1}\underline{\sigma}^2}{2y}, \quad \forall(t,\omega),
$$
the comparison result of Lemma \ref{Lec} leads to
$y_t\geq x_t$ which further yields the lower bound on the maximal solution for $G_t$.

If we define $g_t=\sqrt{G_t}$ and assume that there are two uniformly bounded solutions $g^1(t)$ and $g^2(t)$, then the difference between the two solutions, $\Delta_t=g_t^1-g_t^2$, satisfies
$$\Delta_t=\mathrm{E}\left[\left.\int_{t}^{T}-\Delta_s\frac{\varepsilon^{2H-1}\sigma_s^2}{2g_s^1g_s^2}ds\right|\mathcal{F}_t^\sigma\right].$$
If we denote
$$a_t=\frac{\varepsilon^{2H-1}\sigma_s^2}{2\sqrt{g_s^1g_s^2}},$$
then $a_t\ge 0$ and
$\exp\left(-\int_0^ta_sds\right)\Delta_t$
is a bounded continuous martingale. Therefore, we can finally reach the conclusion that $\Delta_t=0$ since $\Delta_T=g_T^1-g_T^2=0$. This completes the proof.
\end{proof}

\begin{lemma}\label{L43}
The price process $P_t$ driven by \eqref{Pt2} and \eqref{lambdat} converges almost surely to $v$ at time $T$.
\end{lemma}
\begin{proof}
It is straightforward that
\begin{eqnarray}\label{for19}
dP_t &=&\lambda_tdY_t\nonumber\\
&=&\sqrt{\frac{\Sigma_t}{G_t}}\theta_tdt+\sigma_t dB^{\varepsilon,H}_t\nonumber\\
&=&\sqrt{\frac{\Sigma_t}{G_t}}\left[\left(\theta_t+(H-\frac{1}{2}\psi_t)\right)dt+\varepsilon^{H-\frac{1}{2}}\sigma_t dW_t\right]\nonumber\\
&=&\sqrt{\frac{\Sigma_t}{G_t}}\left[\beta_t(v-P_t)dt+\varepsilon^{H-\frac{1}{2}}\sigma_t dW_t\right]\nonumber\\
&=&\frac{v-P_t}{G_t}\varepsilon^{2H-1}\sigma_t^2dt+\sqrt{\frac{\Sigma_t}{G_t}}\varepsilon^{H-\frac{1}{2}}\sigma_t dW_t
\end{eqnarray}
and
\begin{equation}\label{e20}
d\Sigma_t=-\frac{\Sigma_t}{G_t}\varepsilon^{2H-1}\sigma_t^2dt, \quad \Sigma_T=0.
\end{equation}
If we consider the process $X_t=P_t-v$, then
$$X_t=e^{-\int_{0}^{t}\frac{\varepsilon^{2H-1}\sigma_u^2}{G_u}du}X_0+\int_{0}^{t}e^{-\int_{s}^{t}\frac{\varepsilon^{2H-1}\sigma_u^2}{G_u}du}
\sqrt{\frac{\Sigma_s}{G_s}}\varepsilon^{H-\frac{1}{2}}\sigma_s dW_s\triangleq I_1+I_2.$$
From \eqref{bound1}, it is not difficult to figure out
$$\varepsilon^{2H-1}\frac{\underline{\sigma}^2}{\overline{\sigma}^2}\log\left(\frac{T}{T-t}\right)
\leq
\int_{0}^{t}\frac{\varepsilon^{2H-1}\sigma_u^2}{G_u}du.
\leq
\varepsilon^{2H-1}\frac{\overline{\sigma}^2}{\underline{\sigma}^2}\log\left(\frac{T}{T-t}\right),
$$
which directly leads to $\lim_{t\rightarrow T}I_1(t)=0$. Moreover, $I_2(t)$ can be alternatively expressed as
$$I_2(t)=e^{-\int_{0}^{t}\frac{\varepsilon^{2H-1}\sigma_u^2}{G_u}du}M_t,$$
where
$$M_t=\int_{0}^{t}e^{\int_{0}^{s}\frac{\varepsilon^{2H-1}\sigma_u^2}{G_u}du}\sqrt{\frac{\Sigma_s}{G_s}}\varepsilon^{H-\frac{1}{2}}\sigma_sdZ_s$$
is a Brownian martingale, whose quadratic variation is equal to
\begin{equation*}
\begin{aligned}
\langle M\rangle_t&=\int_{0}^{t}e^{\int_{0}^{s}\frac{2\varepsilon^{2H-1}\sigma_u^2}{G_u}du}\frac{\Sigma_s}{G_s}\varepsilon^{2H-1}\sigma^2_sds\\
&=\Sigma_0\int_{0}^{t}e^{\int_{0}^{s}\frac{\varepsilon^{2H-1}\sigma_u^2}{G_u}du}\frac{\varepsilon^{2H-1}\sigma^2_s}{G_s}ds\\
&=\Sigma_0\left(e^{\int_{0}^{t}\frac{\varepsilon^{2H-1}\sigma_u^2}{G_u}du}-1\right).
\end{aligned}
\end{equation*}
According to Lemma \ref{strong0}, there exists a standard Brownian motion $B_t$ such that the continuous martingale can be viewed as a time-changed Brownian motion, i.e., $M_t=B_{\langle M\rangle_t}$. Applying the strong law of large numbers for Brownian motion specified in Lemma \ref{strong}, we finally arrive at the desired result
\begin{equation}
\lim_{t\rightarrow T}I_2(t)=\lim_{t\rightarrow T}e^{-\int_{0}^{t}\frac{\varepsilon^{2H-1}\sigma_u^2}{G_u}du}M_t=
\lim_{t\rightarrow T}\frac{B_{\langle M\rangle_t}}{1+\frac{\langle M\rangle_t}{\Sigma_0}}
=\lim_{\tau \rightarrow \infty}\frac{\frac{B_\tau}{\tau}}{\frac{1}{\Sigma_0}+\frac{1}{\tau}}=0.
\end{equation}
This completes the proof.
\end{proof}

\begin{remark}\label{rre}
If we define the mean-reversion rate of $P_t$ as
$$\kappa_t=\frac{\varepsilon^{2H-1}\sigma_t^2}{G_t},$$
then we can obtain the mean-reverting form of $P_t$ as
\begin{equation}\label{dynamicP}
dP_t=\kappa_t(v-P_t)dt+\sqrt{\Sigma_0}e^{-\int_{0}^{t}\frac{\kappa_s}{2}ds}\sqrt{\kappa_t}dW_t
\end{equation}
and the insider would adopt the following memorable trading strategy
\begin{equation}\label{stra}
\theta_t=\frac{\kappa_t}{\lambda_t}(v-P_t)-(H-\frac{1}{2})\psi_t.
\end{equation}
\end{remark}
\begin{remark}\label{R41}
There is also a useful result about the limiting distribution of the standard price process $h_t=\frac{P_t-v}{\sqrt{\Sigma_t}}$. It\^{o}'s lemma can yield
\begin{equation}
dh_t=-\frac{1}{2}\frac{\varepsilon^{2H-1}\sigma_t^2}{G_t}h_tdt+\frac{\varepsilon^{H-\frac{1}{2}}\sigma_t}{\sqrt{G_t}}dZ_t,
\end{equation}
which implies that $h_t$ is a time-change Ornstein-Uhlenbeck (O-U) process with the following stochastic time change process:
$$
\tau_t=\int_{0}^{t}\frac{\varepsilon^{2H-1}\sigma_s^2}{G_s}ds,
$$
being independent of the filtration generated by $Z_t$. Furthermore, since $E[h_T]=0$ and $E[h^2_T]=1$, the limiting distribution of $h_T$ is a standard normal distribution.
\end{remark}

With the results presented above, we are now able to show that the market depth process is a martingale and a new bound can be established for $G_t$, the details of which are illustrated below.

\begin{lemma}\label{L44}
Market depth process $\frac{1}{\lambda_t}$ is a martingale which is orthogonal to the order flow. Moreover, the price impact process $\lambda_t$ is a submartingale.
\end{lemma}
\begin{proof}
From the definition of $G_t$ specified in Lemma \ref{G001}, it is straightforward to deduce
\begin{equation}
d\sqrt{G_t}+\frac{\varepsilon^{2H-1}\sigma_t^2}{2\sqrt{G_t}}dt=d\mathcal{M}_t,
\end{equation}
where
$$
\mathcal{M}_t=\mathrm{E}\left[\left.\int_{0}^{T}\frac{\varepsilon^{2H-1}\sigma_t^2}{2\sqrt{G_t}}dt\right|\sigma^t\right].
$$
Since
$$
\mathcal{M}_t\leq \frac{\varepsilon^{H-\frac{1}{2}}\overline{\sigma}^2}{\underline{\sigma}}\sqrt{T},
$$
which is a direct result of \eqref{bound1}, $\mathcal{M}_t$ is actually a martingale adapted to the filtration generated by the noise trading volatility process. Using the definition of market depth process, one can easily obtain
\begin{equation}
d\frac{1}{\lambda_t}=d\sqrt{\frac{\Sigma_t}{G_t}}=\frac{1}{\Sigma_t}d\sqrt{G_t}-\frac{\sqrt{G_t}}{2\Sigma_t^{\frac{3}{2}}}d\Sigma_t=\frac{1}{\Sigma_t}d\mathcal{M}_t,
\end{equation}
which clearly shows that $\frac{1}{\lambda_t}$ is a martingale. Furthermore, since $Z_t$ and $M_t$ are independent with each other, we have $d\mathcal{M}_t dZ_t=0$ and so $d\frac{1}{\lambda_t}dY_t=0$. Finally, a further computation using Jensen's inequality yields
$$\frac{1}{\mathrm{E}[\lambda_s|\mathcal{F}_t]}\leq \mathrm{E}\left[\left.\frac{1}{\lambda_s}\right|\mathcal{F}_t\right]=\frac{1}{\lambda_t},$$
indicating that $\lambda_t\leq E[\lambda_s|\mathcal{F}_t]$. This completes the proof.
\end{proof}

\begin{remark}
It should be remarked that this contradicts to the results from most of the previous literature. While the price impact is constant in the original the Kyle model and either a martingale or a supermartingale in various well-known extensions of the Kyle model \cite{Back92,Back98,Back04,Baruch02,Caldentey10}, our result is consistent with what is presented in \cite{Collin16}, where the insider trader has an opportunity to wait for a better liquidity to trade. Indeed, the price impact must increase on average to encourage the insider to trade early and give up his opportunity to wait for better liquidity states under the framework of stochastic liquidity.
\end{remark}

\begin{lemma}
If a bounded solution $G_t$ exists, then
$$
G_t\leq \varepsilon^{2H-1}\mathrm{E}\left[\int_{t}^{T}\sigma_s^2ds\right].
$$
\end{lemma}
\begin{proof}
Applying It\^o's formula to $G_t$ directly yields
\begin{eqnarray}\label{G002}
d\sqrt{G_t}^2&=&2\sqrt{G_t}d\sqrt{G_t}+d[\sqrt{G_t}]_t \nonumber \\
&=&-\varepsilon^{2H-1}\sigma_t^2dt+2\sqrt{G_t}d\mathcal{M}_t+d[\sqrt{G_t}]_t \nonumber \\
&=&-\varepsilon^{2H-1}\sigma_t^2dt+2\sqrt{G_t}d\mathcal{M}_t+\Sigma_t d\left[\frac{1}{\lambda}\right]_t.
\end{eqnarray}
Integrating \eqref{G002} from $t$ to $T$ and taking the expectation on both sides of the resulted equation, we can get the desired result.
\end{proof}
\subsection{Optimal strategy and verifying theorem}

With all necessary results presented in the previous subsection, we are now ready to show that the strategy \eqref{stra} is indeed the optimal solution to the target problem \eqref{max}. The main results are summarized in the following proposition.

\begin{proposition}\label{prop1}
If price dynamic is given by \eqref{Pt}, \eqref{lambdat} and the volatility is uniformly bounded above by $\overline{\sigma}$ and below by $\underline{\sigma}>0$, then the optimal value process defined in \eqref{max} can be derived as
\begin{equation}\label{formJ}
J_t=\frac{(v-P_t)^2+\Sigma_t}{2\lambda_t}+\left(H-\frac{1}{2}\right)\int_{0}^{t}\psi_s(v-P_s) ds-\left(H-\frac{1}{2}\right)A,
\end{equation}
and the optimal trading strategy has the following expression
\begin{equation}\label{e55}
\theta^{*}_t=\frac{\varepsilon^{2H-1} \sigma_t^2}{\lambda_t G_t}(v-P_t)-(H-\frac{1}{2})\psi_t.
\end{equation}
Here, the constant $A$ can be computed from
\begin{equation}\label{assum}
A=\mathrm{E}\left[\left.\int_{0}^{T}\psi_t(v-P_t) dt \right| \mathcal{F}^Y_0,v\right],
\end{equation}
where the dynamic of $P_t$ is specified in \eqref{dynamicP}.
\end{proposition}
\begin{proof}

We conjecture that the target value function can be expressed as \eqref{formJ}. Applying It\^o's lemma to $J_t$ yields
\begin{eqnarray}\label{dJ}
dJ_t&=&\frac{(v-P_t)^2+\Sigma_t}{2}d\lambda_t+\frac{1}{\lambda_t}\left(-(v-P_t)dP_t+\frac{1}{2}(dP_t)^2\right)+\frac{1}{2\lambda_t}d\Sigma_t+\left(H-\frac{1}{2}\right)\psi_t(v-P_t) \nonumber \\
&=&\frac{(v-P_t)^2+\Sigma_t}{2\sqrt{\Sigma_t}}d\mathcal{M}_t+\frac{1}{2\lambda_t}d\Sigma_t + \frac{1}{\lambda_t}\left[-(v-P_t)\lambda_t\left(\theta_tdt+\varepsilon^{H-\frac{1}{2}}\sigma_t dW_t\right)+\frac{1}{2}\varepsilon^{2H-1}\sigma_t^2dt\right] \nonumber\\
&=&\frac{(v-P_t)^2+\Sigma_t}{2\sqrt{\Sigma_t}}d\mathcal{M}_t-\theta_t(v-P_t)dt-\varepsilon^{H-\frac{1}{2}}\sigma_t(v-P_t) dW_t,
\end{eqnarray}
where the second equality is obtained using \eqref{for19} in Lemma \ref{L43} and the third equality is a consequence of \eqref{Sigma9} in Lemma \ref{Ltheta}.
Integrating \eqref{dJ} from $0$ to $T$, we further obtain
\begin{equation}\label{e34}
J_T-J_0+\int_{0}^{T}\theta_t(v-P_t)dt
=\int_{0}^{T}\frac{(v-P_t)^2+\Sigma_t}{2\sqrt{\Sigma_t}}d\mathcal{M}_t+\int_{0}^{T}\varepsilon^{H-\frac{1}{2}}\sigma_t(v-P_t) dW_t.
\end{equation}
If we define
$$I_1(t)=\int_{0}^{t}\varepsilon^{H-\frac{1}{2}}\sigma_s(v-P_s)dW_s,$$
and
$$I_2(t)=\int_{0}^{t}\frac{(v-P_s)^2+\Sigma_s}{2\sqrt{\Sigma_s}}d\mathcal{M}_s,$$
then they are actually martingales for any admissible strategy. To prove this, we start from \eqref{dynamicP} and \eqref{stra} by expressing $P_t$ as
$$P_t=P_0+\int_{0}^{t}\lambda_s\widehat{\theta}_sds+\int_{0}^{t}\varepsilon^{H-\frac{1}{2}}\sigma_s\lambda_sdW_s$$
with $\widehat{\theta}_s=\theta_s+\left(H-\frac{1}{2}\right)\psi_s$. Clearly, one has
$$\int_{0}^{t}\varepsilon^{2H-1}\sigma^2_s\lambda^2_sdW_s=\Sigma_0-\Sigma_t<+\infty$$
and
\begin{eqnarray}
\mathrm{E}\left[\int_{0}^{t}\lambda_s\widehat{\theta}^2_sds\right]\nonumber
&\leq& \mathrm{E}\left[\int_{0}^{t}\lambda^2_sds \int_{0}^{t}\widehat{\theta}^2_sds\right]
\leq \mathrm{E}\left[\int_{0}^{t}\frac{1}{\underline{\sigma}^2}\frac{\Sigma_s}{G_s}\sigma_s^2 ds \int_{0}^{t}\widehat{\theta}^2_sds\right]
=\mathrm{E}\left[\frac{\Sigma_0-\Sigma_t}{\varepsilon^{2H-1}\underline{\sigma}^2}\int_{0}^{t}\widehat{\theta}^2_sds\right]\\\nonumber
&\leq&\mathrm{E}\left[\frac{\Sigma_0}{\varepsilon^{2H-1}\underline{\sigma}^2}\int_{0}^{t}\widehat{\theta}^2_sds\right]
\leq\frac{2\Sigma_0}{\varepsilon^{2H-1}\underline{\sigma}^2}\mathrm{E}\left[\int_{0}^{t}\widehat{\theta}^2_sds\right]\\\nonumber
&\leq&\frac{4\Sigma_0}{\varepsilon^{2H-1}\underline{\sigma}^2}\mathrm{E}\left[\int_{0}^{t}\theta^2_sds+(H-\frac{1}{2})^2\int_{0}^{t}\psi^2_sds\right]
<\infty,
\end{eqnarray}
where the third equality follows from \eqref{e20} and the last inequality is a result of \eqref{e2} and \eqref{e12}. This implies that $P_t$ has finite variance. Considering that $\sigma_t$ is uniformly bounded, we can certainly obtain
$$
\mathrm{E}\left[\int_{0}^{T}\varepsilon^{2H-1}\sigma^2_t(v-P_t)^2dt\right]<\infty,
$$
and thus $I_1(t)$ is a martingale. On the other hand, it follows from Lemma \ref{L44} and Remark \ref{R41} that
$$
\mathrm{E}\left[\frac{(v-P_t)^4}{\Sigma_t}\right]<\Sigma_0 \mathrm{E}[h_t^4]<\infty,
$$
which further leads to
$$\mathrm{E}\left[\int_{0}^{T}\frac{(v-P_s)^4}{\Sigma_s}d\langle\mathcal{M}_s\rangle\right]<\infty.$$
In this case,
$$
\int_{0}^{T}\frac{(v-P_t)^2}{\sqrt{\Sigma_t}}d\mathcal{M}_t
$$
is also a martingale, since $\mathcal{M}_t$ is a uniformly bounded martingale. As a result, considering $\Sigma_t$ is a decreasing process, $I_2(t)$ is a martingale.

Now, if we take expectation on both sides of \eqref{e34}, then
\begin{equation}\label{ine1}
J_0=\mathrm{E}\left[\int_{0}^{T}\theta_t(v-P_t)dt+J_T\right]
\end{equation}
for any $\theta\in\mathcal{A}$. Since $\frac{(v-P_t)^2+\Sigma_t}{2\lambda_t}\geq 0$, the expectation of $J_T$, $\mathrm{E}[J_T]$, can be calculated from \eqref{formJ} by setting $t=T$ and it is straightforward that $\mathrm{E}[J_T]\geq 0$, yielding
$$\mathrm{E}\left[\int_{0}^{T}\theta_t(v-P_t)dt\right]\leq J_0.$$
This implies that $J_0$ will be the optimal value function if there exists a trading strategy $\theta_t^*$, being consistent with \eqref{lambdat}, such that $\mathrm{E}[J_T]=0$. Indeed, from Lemma \ref{G001} and \eqref{assum}, we get
\begin{equation}
\mathrm{E}[J_T]=\mathrm{E}\left[\frac{(v-P_T)^2}{2\lambda_T}+\frac{\sqrt{\Sigma_TG_T}}{2}\right]
=\mathrm{E}\left[\frac{(v-P_T)^2}{2\lambda_T}\right]
\leq \sqrt{\mathrm{E}\left[(v-P_T)^2G_T\right]\mathrm{E}\left[\left(\frac{v-P_T}{\sqrt{\Sigma_T}}\right)^2\right]}
=0.
\end{equation}
Therefore, we have proved the optimality of the value function and that of the trading strategy.
\end{proof}

From Proposition \ref{prop1} and Remark \ref{R41}, it is clear that the equilibrium price admits a bridge process that converges to the value $v$ which is only known to the insider at maturity $T$. This guarantees that all private information will have been incorporated into equilibrium prices at maturity, and the result presented in \cite{Back92} that the equilibrium price in the continuous-time Kyle model follows a standard Brownian Bridge with the constant volatility is generalized. It can also be observed from our model that only if the mean-reversion rate $\kappa_t$ is stochastic will the equilibrium price volatility be stochastic.

Proposition \ref{prop1} further indicates that while the optimal trading strategy for the insider is to trade proportionally to the undervaluation of the asset $v-P_t$ at a rate that is inversely related to his/her price impact $\lambda_t$, it is a monotonic increasing function of the current ``state of liquidity", which is measured by the relative difference between the current noise trading variance $\sigma_t^2$ and the expected noise trading variance $G_t$, i.e., $\frac{\sigma_t^2}{G_t}$. It should also be remarked that due to the presence of the fractional Brownian motion, our optimal strategy $\theta^*(t)$, embracing the stochastic term $(H-\frac{1}{2})\psi_t$, also shows long memory, a property that is not demonstrated in \cite{Collin16}, implying that the introduction of the fractional Brownian motion has a significant impact on the choice of the optimal strategy.

Some further remarks are made below for aggregate execution or slippage costs incurred by uninformed liquidity traders.

\begin{remark}
The total losses between $0$ and $T$ by noise traders can be derived through
\begin{equation}
\int_{0}^{T}(P_{t+dt}-v)\sigma_tdB^{\varepsilon,H}_t=\int_{0}^{T}(P_t+dP_t-v)\sigma_tdB^{\varepsilon,H}_t
=\int_{0}^{T}\varepsilon^{2H-1}\lambda_t\sigma_t^2dt+\int_{0}^{T}(P_t-v)\sigma_tdB^{\varepsilon,H}_t.
\end{equation}
The first component is the pure execution or slippage cost caused by the situation that market orders submitted at time $t$ in the Kyle model will get executed at date $t + dt$ at a price set by competitive market makers. The second component is a fundamental loss resulted from noise traders purchasing a security with long memory, whose fundamental value $v$ is unknown to them. From this, using a similar definition in \cite{Collin16}, aggregate execution or slippage costs incurred by uninformed liquidity traders here can be obtained as
\begin{equation}
\int_{0}^{T}\sigma_tdB^{\varepsilon,H}_tdP_t=\int_{0}^{T}\varepsilon^{2H-1}\lambda_t\sigma_t^2dt,
\end{equation}
which is stochastic, being path-dependent, and is affected by the fractional noise.

\end{remark}
\begin{remark}
The unconditional expected profits of the insider can be determined through
\begin{equation*}
\begin{aligned}
\mathrm{E}\left[\int_{0}^{T}(v-P_t)\theta_t dt\right]
&=\mathrm{E}\left[\int_{0}^{T}\frac{\varepsilon^{2H-1}\sigma_t^2}{\sqrt{\Sigma_tG_t}}(v-P_t)^2 dt \right]-(H-\frac{1}{2})\mathrm{E}\left[\int_{0}^{T}(v-P_t)\psi_tdt \right]\\
&=\mathrm{E}\left[\int_{0}^{T}\frac{\varepsilon^{2H-1}\sigma_t^2}{\sqrt{\Sigma_tG_t}}\Sigma_t dt\right]-(H-\frac{1}{2})\mathrm{E}\left[\int_{0}^{T}(v-P_t)\psi_tdt \right]\\
&=\mathrm{E}\left[\int_{0}^{T}\varepsilon^{2H-1}\lambda_t\sigma_t^2dt\right]-(H-\frac{1}{2})\mathrm{E}\left[\int_{0}^{T}(v-P_t)\psi_tdt \right],
\end{aligned}
\end{equation*}
where the first equality is obtained with the substitution of $\theta_t^*$ and the second one is derived using the law of iterated expectations. Being different from the results in \cite{Collin16}, the unconditional expected profits of the insider is no longer equal to the unconditional expected execution costs paid by noise traders, and instead, they have an additional component due to the introduction of the fractional Brownian motion, implying that the property of long memory possessed by the aggregate order flow has an influence on the insider's unconditional expected profits.
\end{remark}

\section{Trading volume, volatility and price impact}

In this section, two special cases are considered to further investigate the effect of the noise trading volatility processes on the equilibrium, and they are distinguished by whether the growth rate of noise trading is stochastic

\subsection{Deterministic growth rate of the noise trading volatility}
This subsection will discuss the case where the growth rate of the noise trading volatility process in \eqref{MMT} is deterministic (the volatility of that takes a general form). Under this particular assumption, a closed-form solution for $G_t$ can be derived, based on which the equilibrium price process, the equilibrium trading strategy, the equilibrium volatility, and the equilibrium price impact can be obtained. The corresponding results are presented in the following proposition.

\begin{proposition}
Suppose that the growth rate of the noise trading volatility is deterministic such that
\begin{equation}
D_t=\int_{t}^{T}e^{\int_{t}^{u}2m_sds}du
\end{equation}
is bounded for all $t \in [0,T]$. Then the solution to \eqref{GGT} can be expressed as
\begin{equation}\label{GGT2}
G_t=\varepsilon^{2H-1}\sigma_t^2D_t
\end{equation}
and the stock price dynamic has the following form
\begin{equation}
dP_t=\frac{1}{D_t}(v-P_t)dt+e^{\int_{0}^{t}m_sds}\sigma_vdW_t,
\end{equation}
where $\sigma_v^2=\frac{\Sigma_0}{D_0}$. In equilibrium, the price impact can be represented by
\begin{equation}
\lambda_t=e^{\int_{0}^{t}m_sds}\frac{1}{\varepsilon^{H-\frac{1}{2}}}\frac{\sigma_v}{\sigma_t}
\end{equation}
and the optimal trading strategy of the insider can be formulated as
\begin{equation}
\theta^*_t=\frac{1}{\lambda_tD_t}(v-P_t)-(H-\frac{1}{2})\psi_t.
\end{equation}
Furthermore, the expected trading rate of the insider is
\begin{equation}
\mathrm{E}[\theta|v,\mathcal{F}_0]=e^{\int_{0}^{t}2m_sds}\frac{\sigma_0(v-P_0)}{\sigma_vD_0}.
\end{equation}
\end{proposition}

\begin{proof}
With the utilization the martingale property, it is straightforward from \eqref{MMT} that
\begin{equation}
\mathrm{E}\left[\sigma_u|\mathcal{F}_t\right]=\sigma_te^{\int_{t}^{u}m_sds}.
\end{equation}
In this case, if we assume that the solution to $G_t$ takes the form of \eqref{GGT2}, then it follows from \eqref{GGT} that
$$\sqrt{D_t}=\int_{t}^{T}\frac{e^{\int_{t}^{u}m_sds}}{2\sqrt{D_u}}du.$$
Clearly, our guess is correct if $D_t$ specified in the proposition satisfies this integral equation, which is exactly the case here. Considering the uniqueness we established above, \eqref{GGT2} is indeed the expression of the target solution.

The expected trading rate of the insider can be computed from
\begin{equation}
\mathrm{E}[\theta_t|v,\mathcal{F}_0]=E\left[\frac{v-P_t}{\sqrt{\Sigma_t}}\frac{\varepsilon^{2H-1}\sigma_t^2}{\sqrt{G_t}}-(H-\frac{1}{2})\psi_t\right]=E\left[\frac{v-P_t}{\sqrt{\Sigma_t}}\frac{\varepsilon^{2H-1}\sigma_t^2}{\sqrt{G_t}}\right]=\frac{v-P_0}{\sqrt{\Sigma_0}}e^{-\int_{0}^{t}\frac{1}{2D_s}ds}\frac{\sigma_0\varepsilon^{H-\frac{1}{2}}e^{\int_{0}^{t}m_sds}}{\sqrt{D_t}},
\end{equation}
which is a direct result of \eqref{GGT}, the dynamic of $h_t$ specified in Lemma \ref{L44}, and $\psi_t$ being a martingale. Thus, using the identity
$$e^{-\int_{0}^{t}\frac{1}{2D_s}ds}=\sqrt{\frac{D_t}{D_0}}e^{-\int_{0}^{t}m_sds},$$
directly yields the desired result.

The other results in the proposition actually follow from Proposition \ref{prop1}, and this has completed the proof.
\end{proof}

It should be remarked that our model has successfully taken into consideration the effect of long memory, after the introduction of the fractional Brownian motion. Of course, this model takes \cite{Collin16} as a special case when $H=\frac{1}{2}$ is set in\eqref{for19}, and it will further degenerate to the continuous-time Kyle model \cite{Back92} when $\sigma_t=\sigma$ with $m_t=v_t=0$ and $D_t=T-t$. In this case, with $\sigma_v^2=\frac{\Sigma_0}{T}$ being the annualized variance of the market maker's prior, both of the price volatility and price impact are constant, being equal to $\sigma_v$ and $\frac{\sigma_v}{\sigma}$, respectively.

It should also be noted that the price volatility and the posterior variance of the fundamental value $\Sigma_t$ are deterministic, as a result of the growth rate of the noise trading volatility being deterministic, while the price impact is stochastic and negatively correlated with the noise trading volatility. On the other hand, the optimal trading trading strategy of the insider is not only negatively and positively dependent on the price impact $\lambda_t$ and the liquidity state $\kappa_t$, respectively, it also exhibits the property of long memory, due to the introduction of the fractional Brownian motion. It is also interesting to notice that the fractional noise has no influence on the expected trading rate of the insider, which is also reasonable as this is in the sense of average.

\begin{figure}[H]\centering
\includegraphics[width=.5\textwidth]{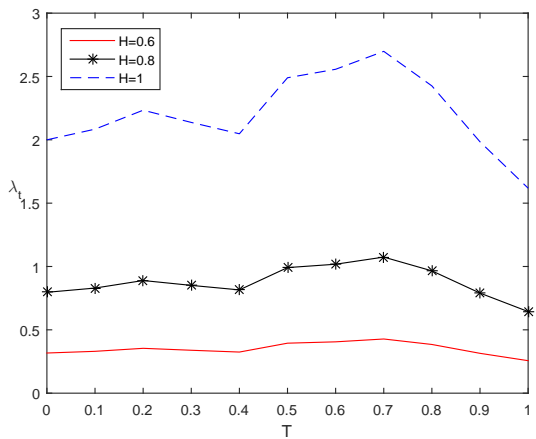}
\caption{The impact of $H$ on $\lambda_t$. Parameter values are: $\Sigma_0=0.2^2$, $\sigma_0=1$, $m_s=1$, T=1, $\varepsilon=0.01$.}
  \label{fig2}
\end{figure}

\begin{figure}[H]\centering
\includegraphics[width=.5\textwidth]{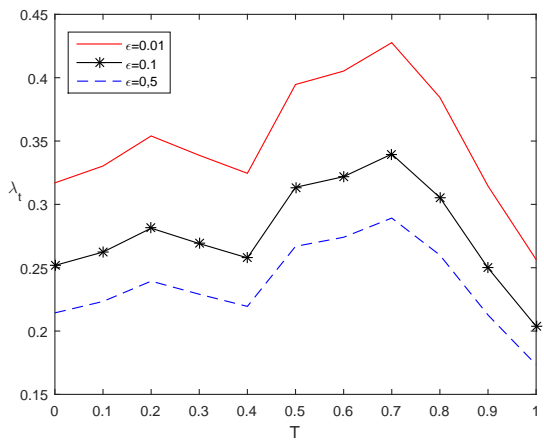}
\caption{The impact of $\varepsilon$ on $\lambda_t$. Parameter values are: $\Sigma_0=0.2^2$, $\sigma_0=1$, $m_s=1$, T=1, $H=0.6$.}
  \label{fig1}
\end{figure}

To further investigate the effect of long memory, how the price impact $\lambda_t$ (Kyle's lambda) changes with respect to the long-memory parameter $H$ and the approximation factor $\varepsilon$ is numeircally shown. In specific, Fig. \ref{fig2} displays that the long-memory parameter $H$ has a positive influence on the price impact, as a higher $H$ value contributes to a higher price impact. This can be understood by the fact that more information is observed from the order flow when the volatility shows greater long range dependence. Moreover, the price impact is more sensitive to the time when $H$ increases, while it is almost a constant when $H$ is low, which is expected as the approximate fractional Brownian motion degenerates to the standard Brownian motion when $H$ approaches $\frac{1}{2}$ and the long memory property no longer exists. On the other hand, an apposite trend can be observed in Fig. \ref{fig1} that the smaller the approximation factor $\varepsilon$ is, the larger the price impact will be. This is also reasonable since there will be less fractional noise when $\varepsilon$ decreases, implying that more information will be revealed, leading to a larger price impact.

However, it needs to be stressed that although the price impact is stochastic, the price volatility is still a deterministic function, and neither contemporaneous relation between volume changes and the price volatility nor that between the price impact and price volatility can be generated. Such relations can only be generated when the growth rate of the noise trading volatility is stochastic. In the next subsection, a framework generating both stochastic price volatility and a meaningful correlation between the price volatility and volume will be introduced.

\subsection{Stochastic growth rate of the noise trading volatility}

In this subsection, a special case where the noise trading volatility follows a two-state continuous Markov chain is presented, and this introduces state-dependent predictability, resulting in successfully capturing the stochastic expected growth rate in the noise trading volatility.

We start by specifying the dynamic of $\sigma_t$ as
\begin{equation}
d\sigma_t=\left(\sigma^H-\sigma_t\right)dN_L(t)-(\sigma_t-\sigma^L)dN_H(t),
\end{equation}
where $\sigma^L$ and $\sigma^H$ are two fixed values satisfying $\sigma^L<\sigma^H$, the initial level of the volatility $\sigma_0\in \{\sigma^H,\sigma^L\}$, and $N_i(t)$ is a standard Poisson counting process with jump intensity $\lambda_i$ for $i=H$, $L$. In this case, the solution to \eqref{GGT} is presented in the following proposition.

\begin{proposition}
The unique bounded solution to \eqref{GGT} is given by
$$G_t=\textbf{1}_{\{\sigma_t=\sigma^L\}}\varepsilon^{2H-1}G^L(T-t)+\textbf{1}_{\{\sigma_t=\sigma^H\}}\varepsilon^{2H-1}G^H(T-t),$$
where $G^L$ and $G^H$ are the solutions to the following ODEs (ordinary differential equations)
\begin{equation}\label{JJ}
\left\{
\begin{aligned}
dG^L(\tau)&=(\sigma_L)^2+2\lambda_L\left(\sqrt{G^H(\tau)G^L(\tau)}-G^L(\tau)\right),\\
dG^H(\tau)&=(\sigma_H)^2+2\lambda_H\left(\sqrt{G^H(\tau)G^L(\tau)}-G^H(\tau)\right)
\end{aligned}
\right.
\end{equation}
with the boundary conditions $G^L(0)=G^H(0)=0$.
\end{proposition}

\begin{proof}
If we define
$$G(t,\sigma_t)=\textbf{1}_{\{\sigma_t=\sigma^L\}}G^L(T-t)+\textbf{1}_{\{\sigma_t=\sigma^H\}}G^H(T-t)$$
and
$$M(t)=\sqrt{G(t,\sigma_t)}+\int_{0}^{t}\frac{\sigma_u^2}{2\sqrt{G(u,\sigma_u)}}du,$$
then it is not difficult to find that $M(t)$ is a pure jump martingale, which implies that $M(t)=\mathrm{E}[M(T)|\mathcal{F}_t]$, leading to
$$\sqrt{G(t,\sigma_t)}+\int_{0}^{t}\frac{\sigma_u^2}{2\sqrt{G(u,\sigma_u)}}du=\mathrm{E}\left[\left.\sqrt{G(T,\sigma_T)}+\int_{0}^{T}\frac{\sigma_u^2}{2\sqrt{G(u,\sigma_u)}}du\right|\mathcal{F}_t\right].$$
Considering the boundary conditions $G^L(0)=G^H(0)=0$, it is straightforward that
$$\sqrt{G(t,\sigma_t)}=\mathrm{E}\left[\left.\int_{t}^{T}\frac{\sigma_u^2}{2\sqrt{G(u,\sigma_u)}}du\right|\mathcal{F}_t\right],$$
and thus
$$\sqrt{\varepsilon^{2H-1}G(t,\sigma_t)}=\varepsilon^{H-\frac{1}{2}}\sqrt{G(t,\sigma_t)}=\mathrm{E}\left[\left.\int_{t}^{T}\frac{\varepsilon^{2H-1}\sigma_u^2}{2\sqrt{\varepsilon^{2H-1}G(u,\sigma_u)}}du\right|\mathcal{F}_t\right].$$
This has completed the proof.
\end{proof}

According to Remark \ref{rre}, the price dynamic can be written as
\begin{equation}
dP_t=\kappa(t,\sigma_t)(v-P_t)dt+\sigma_P(t)dW_t,
\end{equation}
where
$$\sigma_P(t)=\sqrt{\Sigma_0}e^{-\int_{0}^{t}\frac{1}{2}\kappa(s,\sigma_s)ds}\sqrt{\kappa(t,\sigma_t)}$$
and
$$\kappa(t,\sigma_t)=\textbf{1}_{\{\sigma_t=\sigma^L\}}\kappa^L(T-t)+\textbf{1}_{\{\sigma_t=\sigma^H\}}\kappa^H(T-t)$$
with
$$
\kappa^i(T-t)=\frac{(\sigma^i)^2}{G^i(T-t)}, \quad i=L,H.
$$
Clearly, it can be easily observed that the price now follows a mean-reverting process with stochastic volatility, with both of the mean-reversion speed and the volatility being controlled by the Markov chain. Moreover, a high value of the noise trading volatility always contributes a higher mean-reversion speed, which is because the insider is expected to trade more aggressively. It should also be noted that both the mean-reversion speed and the volatility would approach infinity when the time becomes closer to expiry, which can be explained by the fact that agents trade more and more aggressively when it is approaching expiry since they do not want to leave any money on the table. A positive relationship between volume changes and the price volatility can also be witnessed, as the price volatility always jumps in the same direction as the noise trading volatility does.

\subsection{Extensions and discussions}

While all the discussions above are based on the assumption that the aggregate order flow and noise trading volatility are conditionally uncorrelated for the simplicity of the illustration, this assumption can in fact be relaxed. To illustrate this, let us consider a more general model with the total order flow $Y'_t$ being defined as follows:
\begin{equation}
dY^*_t=\theta_tdt+\sigma_tdB^{\varepsilon,H}_t+\eta(t,\sigma^t,Y^*_t)dM_t.
\end{equation}
If we further define
\begin{equation}\label{E36}
dY_t=dY^*_t-\frac{\eta(t,\sigma^t,Y^*_t)}{v(t,\sigma^t)}(d\sigma_t-m(t,\sigma^t)dt)=\eta_tdt+\sigma_tdB^{\varepsilon,H}_t,
\end{equation}
which is exactly the same as what is presented in \eqref{E36}, then it can be easily shown that all our results above are unchanged under this generalized framework, since observing $(Y^{*}_{t},\sigma_t)$ is equivalent to observing $Y_t,\sigma_t$ for all market participants. However, one should also notice that our equilibrium proof does depend on the assumption that the history of $Y^t$ has no influence on the future dynamic of $\sigma_t$. On the other hand, since the price change is now linear in $Y_t$, it is no longer linear in the total order flow, and instead it is only linear in the component of the order flow that is informative about the insider's actions.

Inspired by the ideas of \cite{Andersen96,Clark73,Collin16,Gallant92}, we also model the time series of price changes as subordinated to the normal distribution. To capture the heteroscedasticity in returns, microeconomic foundations for such a subordinate process modeling the stock return are provided under our model, with the directing process being endogenous and related to the trading volume, and the results are presented in the following proposition.

\begin{proposition}\label{PP2}
If we define the positive increasing stochastic directing process as
$$\tau_t=T\left(1-e^{-\int_{0}^{t}\frac{\varepsilon^{2H-1}\sigma_u^2}{G_u}du}\right)$$
and assume $\sigma_v=\frac{\Sigma_0}{T}$, then
$$\Sigma_t=\sigma_v^2(\tau_T-\tau_t)$$
and
$$dP_t=\frac{v-P_t}{\tau_T-\tau_t}d\tau_t+\sigma_vdW'_{\tau_t}$$
for some Browinian motion $W'$ independent of $M$, with its definition as $\sigma_vdW'_\tau=\varepsilon^{H-\frac{1}{2}}\lambda_t\sigma_tdW_{\tau_t}$.
\end{proposition}
\begin{proof}
The definition of the time-change yields
$$d\tau_t=-\frac{1}{\sigma_v^2}d\Sigma_t=\frac{\varepsilon^{2H-1}\Sigma_t\sigma_t^2}{\sigma_v^2G_t}=\frac{\varepsilon^{2H-1}(\tau_T-\tau_t)\sigma_t^2}{G_t},$$
the substitution of which into the equilibrium price process along with Lemma \ref{L43} and Remark \ref{R41} leads to the desired result.
\end{proof}

%\begin{remark}
%We note that our model is also related to a large literature in financial econometrics initiated by Clark \cite{Clark73}, which models the price change process as subordinated to the normal distribution, with a directing process related to ``informed volume" usually modeled via some latent 'information variable' extracted from volume and price volatility in literature; for instance, we refer to Andersen \cite{Andersen96}, Tauchen and Pitts \cite{Tauchen83}, Richardson and Smith \cite{Richardson93}. In our extended model, price follows such a subordinate process. Interestingly, the model provides the endogenous directing process as a function of the uninformed trading volume process.
%\end{remark}

From Proposition \ref{PP2}, it is clear that the equilibrium price is a time-changed Brownian Bridge, similar to what is presented in the Kyle-Back model, where the price process is a standard Brownian Bridge. However, our equilibrium can not be derived simply using a time-change of that model, which is a result of the fact that the price impact is constant in the Kyle-Back model, while it is a stochastic process in our case. One may also find that price is a time-changed Brownian motion, belonging to the class of subordinate processes proposed in \cite{Clark73}, in the market maker's filtration. Moreover, our model gives an endogenous expression for the directing process $\tau_t$, which depends on the (uninformed) volume dynamic and the fractional noise, while having no requirement on the specification of a latent information process to generate stochastic volatility \cite{Andersen96}. An obvious advantage of this model is its generality, as the directing process can be determined for any dynamic of volume, typical examples of which include Normal, Poisson, and Log-normal, implying that major stylized facts can be jointly taken into consideration, which is an important property according to \cite{Gallant92}.

\section{Conclusions}
In this paper, we propose a modified the Kyle model for dynamic insider trading, with the noise trading volatility and trading volume being respectively governed by a general stochastic process and a fractional stochastic process. Under equilibrium conditions, the resulted equilibrium price process exhibits excessive volatility because of the insider trading more aggressively when uninformed volume is higher. The optimal insider trading strategy displays long memory, and the price impact is negatively correlated with the noise trading volatility, which is also affected by the fractional noise.

The model makes many simplifying assumptions that could be relaxed to further our standing of how information flows into prices and how price volatility, price impact and trading volume change. For instance, we can consider the investment under varying time horizon instead of fixed horizon and assume that the presence of the insider is common knowledge for the market. We leave these extensions for future research.


\begin{thebibliography}{99}
\bibitem{Admati88}A. Admati, P. Pfleiderer, A theory of intraday patterns: volume and price variability, {\it Rev. Financ. Stud.}, {\bf 1 (1)} (1988), 3-40.
\bibitem{Andersen96}T.G. Andersen, Return volatility and trading volume: an information flow interpretation
of stochastic volatility, {\it J. Finance}, {\bf 51 (1)} (1996), 169-204.
\bibitem{Back92}K. Back, Insider trading in continuous time, {\it Rev. Financ. Stud.}, {\bf 5} (1992), 387-409.
\bibitem{Back98}K. Back, H. Pedersen, Long-lived information and intraday patterns, {\it J. Financ. Mark.}, {\bf 1} (1998), 385-402.
\bibitem{Back04}K. Back, S. Baruch, Information in securities markets: Kyle meets Glosten and
Milgrom, {\it Econometrica}, {\bf 72 (2)} (2004), 433-465.
\bibitem{Baruch02}S. Baruch, Insider trading and risk aversion, {\it J. Financ. Mark.}, {\bf 5 (4)} (2002),
451-464.
\bibitem{Bollerslev99}T. Bollerslev, D. Jubinski, Equity trading volume and volatility: latent information arrivals and common long-run dependencies, {\it J. Bus. Econ. Stat.}, {\bf 17} (1999), 9-21.
\bibitem{Caldentey10}R. Caldentey, E. Stacchetti, Insider trading with a random deadline, Econometrica, {\bf 78 (1)} (2010), 245-283.
\bibitem{Chang05}G. Chang, S. Suresh, Asset prices and default-free term structure in an equilibrium model of default, {\it J. Bus.}, {\bf 78} (2005), 1215-1266.
\bibitem{Clark73}P.K. Clark, A subordinated stochastic process model with finite variance for speculative prices, {\it Econometrica}, {\bf 41 (1)} (1973), 135-155.
\bibitem{Collin16}P. Collin-Dufresne, V. Fos, Insider trading, stochastic liquidity, and equilibrium prices, {\it Econometrica}, {\bf 84 (4)} (2016), 1441-1475.
\bibitem{Diebold01}F. Diebold, A. Inoue, Long memory and regime switching, {\it J. Econometrics}, {\bf 105} (2001), 131-159.
\bibitem{Foster90}F.D. Foster, S. Viswanathan, A theory of the interday variations in volume, variance, and trading costs in securities markets, {\it Rev. Financ. Stud.}, {\bf 3 (4)} (1990), 593-624.
\bibitem{Foster93}F.D. Foster, S. Viswanathan, Variations in trading volume, return volatility, and trading costs: evidence
on recent price formation models, {\it J. Finance}, {\bf 48 (1)} (1993), 187-211.
\bibitem{Foster95}F.D. Foster, S. Viswanathan, Can speculative trading explain the volume-volatility relation? {\it J.
Bus. Econ. Stat.}, {\bf 13 (4)} (1995), 379-396.
\bibitem{Gallant92}A.R. Gallant, P.E. Rossi, G. Tauchen, Stock prices and volume, {\it Rev. Financ. Stud.}, {\bf 5 (2)} (1992), 199-242.
\bibitem{Granger04} C. Granger, N. Hyung, Occasional structural breaks and long memory with an application to the S\&P 500 absolute stock returns, {\it J. Empir. Financ.}, {\bf 11} (2004), 399-421.
\bibitem{Guasoni}P. Guasoni, M.H. Weber, Rebalancing multiple assets with mutual price impact, {\it J. Optim. Theory Appl.}, {\bf 179 (2)} (2018), 618-653.
\bibitem{He16}X.J. He, S.P. Zhu, An analytical approximation formula for European option pricing under a new stochastic volatility model with regime-switching, {\it J. Econom. Dynam. Control.}, {\bf 71} (2016), 77-85.
 \bibitem{Hyung06}N. Hyung, S. Poon, C. Granger, A Source of Long Memory in Volatility. Working Paper, University of California, San Diego, 2006.
\bibitem{Karatzas} I. Karatzas, S.E. Sherve, Brownian Motion and Stochastic Calculs, Berlin: Springer-Verlag, 1991.
\bibitem{Kyle85}A. Kyle, Continuous auctions and insider trading, {\it Econometrica}, {\bf 53 (6)} (1985), 1315-1335.
\bibitem{Naik90} V. Naik, M. Lee, General equilibrium pricing of options on the market portfolio with discontinuous returns. {\it Rev. Financ. Stud.}, {\bf 3} (1990), 493-521.
\bibitem{Lepeltier97}J.P. Lepeltier, J.S. Martin, Existence for BSDE with superlinear-quadratic coefficient, {\it Stoch. Stoch. Rep.}, {\bf 63 (3-4)} (1992), 227-240.
\bibitem{Liptser01}R. Liptser, A. Shiryaev, Statistics of Random Processes II: Applications (Second ed.), Berlin: Springer-Verlag, 2001.
\bibitem{Mandelbrot}B. Mandelbrot, J. van Ness, Fractional Brownian motions, fractional noises and applications, {\it SIAM Rev.}, {\bf 10 (4)} (1968), 422-437.
\bibitem{Mrazek}M. Mr\'azek, J. Posp\'i\v{s}il, T. Sobotka, On calibration of stochastic and fractional stochastic volatility models, {\it Eur. J. Oper. Res.} {\bf 254} (2016), 1036-1046.
\bibitem{Liu15}J. Liu, J. Pan, T. Wang, An equilibrium model of rare-event premiums and its implication for option smirks. {\it Rev. Financ. Stud.}, {\bf 8} (2005), 131-164.
%\bibitem{Richardson93}M. Richardson, T. Smith, A test for multivariate normality in stock returns, {\it J. Bus.}, {\bf 66 (2)} (1993), 295-321.
\bibitem{Shiryaev99}A.N. Shiryaev, On arbitrage and replication for fractal models, Preprint, Moscow University and Steklov Institute, 1999.
%\bibitem{Tauchen83}G.E. Tauchen, M. Pitts, The price variability-volume relationship on speculative markets, {\it Econometrica}, {\bf 51 (2)} (1983), 485-505.
\bibitem{Thao}T.H. Thao, An approximate approach to fractional analysis for finance. {\it Nonlinear Anal. RW}, {\bf 7} (2006), 124-132.
\bibitem{Yang17}B.Z. Yang, J. Yue, N.J. Huang, Variance swaps under L\'{e}vy process with stochastic volatility and
  stochastic interest rate in incomplete market. arXiv:1712.10105[q-fin.PR].
%\bibitem{Yang18}B.Z. Yang, S.P. Zhu, N.J. Huang, 2018. Equilibrium prices of volatility derivatives under a general multi-factor volatility model with L\'evy jumps. Submitted.
\bibitem{Yue17}J. Yue, N.J. Huang, Neutral and indifference pricing with stochastic correlation and volatility, {\it J. Ind. Manag. Optim}, {\bf 14 (1)} (2018),  199-229.
\bibitem{Yue18}J. Yue, N.J. Huang, Fractional Wishart processes and $\varepsilon$-fractional Wishart processes with applications, {\it Comput. Math. Appl.}, {\bf 75 (8)} (2018), 2955-2977.
\bibitem{Zhu18}S.P. Zhu, X.J. He, A new closed-form formula for pricing European options under a skew Brownian motion, {\it Eur. J. Financ.}, {\bf 24 (12)} (2018), 1063-1074.
\end{thebibliography}
\end{document}